\documentclass[sigconf]{acmart}

\pdfoutput=1

\setcopyright{none}
\renewcommand\footnotetextcopyrightpermission[1]{}
\usepackage{microtype}
\usepackage{graphicx}
\usepackage{fixltx2e}
\usepackage{url}
\usepackage{algpseudocode}
\usepackage{algorithm}
\usepackage{msc}
\usepackage[n, advantage, operators, sets, adversary, landau, probability, notions, logic, ff, mm, primitives, events, complexity, asymptotics, keys]{cryptocode}
\usepackage{enumitem}

\newcommand{\e}{\epsilon}
\newcommand{\Enc}{\textsc{Enc}}
\newcommand{\Dec}{\textsc{Dec}}
\newcommand{\LSB}{\mathrm{LSB}}
\newcommand{\MSB}{\mathrm{MSB}}
\newcommand{\Encode}{\textsc{Encode}}
\newcommand{\Decode}{\textsc{Decode}}

\newcommand{\Gen}{\textsc{Gen}}
\newcommand{\id}{\mathrm{id}}
\newcommand{\Cert}{\mathrm{Cert}}
\newcommand{\usr}{\mathrm{usr}}

\newcommand{\pwd}{\mathrm{pwd}}
\newcommand{\info}{\mathrm{info}}
\newcommand{\xtime}{\mathrm{time}}
\newcommand{\nonce}{\mathrm{nonce}}
\newcommand{\websites}{\mathrm{websites}}
\newcommand{\sqn}{\mathrm{sqn}}

\newcommand{\MAC}{\textsc{MAC}}
\newcommand{\Vrfy}{\textsc{VrfyMAC}}
\newcommand{\MACtag}{\mathrm{tag}}
\newcommand{\PRG}{\textsc{PRG}}
\newcommand{\Sign}{\textsc{Sign}}

\newcommand{\cale}{\mathcal{E}}
\newcommand{\F}{\mathbb{F}}

\newcommand{\IndistExp}{\mathrm{ExpIndDistress}^{f_{N}}_\adv}
\newcommand{\IndDistress}{\mathrm{IndDistress}}
\newcommand{\Adv}{\mathrm{Adv}}

\newcommand{\randomCPA}{\textrm{IND\$-CPA}}

\newcommand{\IndAdv}{\Adv_\adv^{\IndDistress,f_N}}

\newcommand{\advloc}{\adv_{\textsc{loc}}}
\newcommand{\advnet}{\adv_{\textsc{net}}}
\newcommand{\advweb}{\adv_{\textsc{web}}}
\newcommand{\DCP}{\textrm{DCP}}

\newtheorem{guarantee}{Guarantee}
\newtheorem{prop}{Proposition}

\begin{document}

\title{Ask for Alice: Online Covert Distress Signal in the Presence of a Strong Adversary}


\author{Hayyu Imanda}
\affiliation{
	\institution{University of Oxford}
	\city{Oxford}
	\country{United Kingdom}
}

\author{Kasper Rasmussen}
\affiliation{%
 	 \institution{University of Oxford}
  	\city{Oxford}
	\country{United Kingdom}
}
%
%
%
%
%


\begin{abstract}

  In this paper we propose a protocol that can be used to covertly send a distress signal through a seemingly normal webserver, even if the adversary is monitoring both the network and the user's device. This allows a user to call for help even when they are in the same physical space as their adversaries. We model such a scenario by introducing a strong adversary model that captures a high degree of access to the user's device and full control over the network. 
Our model fits into scenarios where a user is under surveillance and wishes to inform a trusted party of the situation. To do this, our method uses existing websites to act as intermediaries between the user and a trusted backend; this enables the user to initiate the distress signal without arousing suspicion, even while being actively monitored. We accomplish this by utilising the TLS handshake to convey additional information; this means that any website wishing to participate can do so with minimal effort and anyone monitoring the traffic will just see common TLS  connections. In order for websites to be willing to host such a functionality the protocol must coexist gracefully with users who use normal TLS and the computational overhead must be minimal. We provide a full security analysis of the architecture and prove that the adversary cannot distinguish between a set of communications which contains a distress call and a normal communication.

\end{abstract}




\maketitle

\section{Introduction}

In the United Kingdom, the \emph{Ask for Angela} campaign~\cite{angela} allows any individual who feels unsafe an exit from the establishment they are in by asking for a fictional member named Angela, prompting a trained member of staff to ensure a safe passage and potentially call the authorities. Even more recently, major pharmacies across the UK have taken part in \emph{Ask for Ani} (Action Needed Immediately)~\cite{ani}, where the pharmacy would provide a consultation room for people experiencing domestic abuse. These campaigns are aimed to allow survivors of domestic abuse to systematically be able to inform someone of their situation in a discreet manner, in a physical space that is common for them to visit.

During the Covid-19 pandemic, there has been an increase in demand for survivors of domestic abuse services, which indicate an increase in the severity of abuse being perpetrated while survivors were not able to leave home~\cite{domesticabusecorona}; EU Member States experienced a similar trend~\cite{Mahasem1872}. While some regions in the US have reported that the number of calls received by domestic violence hotlines have dropped, experts believe that this is not a reflection of a decrease in the number of cases but rather that survivors were unable to safely connect with services while they are in the same space as their abusers~\cite{doi:10.1056/NEJMp2024046}. A solution is needed where survivors can safely reach out from their own homes, and this is a difficult task.

In this paper, we present an adversary model where the adversary has a high degree of control over a user's device and complete control over the local network. Our model allows an abstraction of real-life scenarios, where an adversary has visual access to the user's device while the user wishing to request help---this can be through sharing a physical space or the adversary's ability to monitor the screen through a camera. In addition to the above, this may include travellers who have been required to enter their authentication information at a hostile nation's border~\cite{noauthor_chinese_nodate}, or an intelligence agent who has to covertly inform of their capture while being in full surveillance by their captors~\cite{funkspiel}.


In certain situations, if communication dedicated to signalling distress is detected, it can lead to further consequences from the adversary. Hence, it is important that the system hidesthe fact that the user is signalling distress. We accomplish this using several different techniques: firstly, we rely on (existing) webservers to enrol in the scheme so they can act as intermediaries in communication between the user and the backend---this goes a long way towards making the connection seem innocent even if the device is under observation; secondly, to make sure the distress signal is undetectable on the network, we embed the distress signal by utilising the random nonce in the TLS handshake as a covert channel. This ensures that a signal can be sent without the adversary's knowledge while not compromising the underlying TLS session. We formalise our security notion using a security game, and with the formal definition of security and the adversary model we are able to provide a thorough proof of the security of our protocols. We summarise our contributions as follows:
\begin{itemize}
\item We introduce a local adversary $\advloc$, who in addition to being a Dolev-Yao adversary on the network, also has visual access to the user's device screen as well as all the user's TLS application layer data.
\item We create an infrastructure for distress reporting, where one central entity receives distress signals along with accompanying information, as well as webservers who voluntarily participate acting as entry points.
\item Within the distress infrastructure, we introduce enrolment protocols for the user and the webserver, as well as the main protocol with which the distress signal is relayed by the webserver to the central entity.
\item We introduce encoding and decoding functions to hide the distress and necessary information while being sent to the webserver with the presence of $\advloc$. We piggyback on the TLS handshake to send the output of the encoding function, and we provide a security analysis to show that $\advloc$ will not be able to distinguish between a distribution of nonces containing a distress signal and that of normal communications.
\end{itemize}

The fundamental goal of our scheme is for the user to signal distress to an online entity, where enrolment is possible beforehand. There are many non-technical facets of our scheme that are out of scope for this paper, including the resolution of the situation after the distress is received. Further, our scheme discusses the protocols that occur behind the scenes of the user initiating the distress; the human interaction with the device needs to be designed in a proper, user-friendly manner. These issues are incredibly important but are considered future work, however we highlight some possibilities in Section~\ref{sec:discussion}.

\section{Design Goals}

One trivial attempt to achieve the user's goal of communicating a distress signal is for the user to send an encrypted message to a trusted third party (e.g. a police website). However, the adversary, having full control of the network, is trivially aware of such attempt by the user. However, indeed for a very specific use-case, one webserver can be set up for the specific purpose of receiving a distress signal (i.e. the user communicates directly to the trusted third party)---this is possible only if the adversary isn't aware of the purpose of such website.

A second attempt may gain inspiration from existing schemes: in UK pharmacies participating in the \emph{Ask for Ani} scheme, a domestic abuse survivor may enter the establishment, even with the presence of the abuser, and ask for \emph{Ani} to a staff member who will give them access to a safe space and contact the relevant authorities. In network communications, this translates to the user signalling distress by sending messages that may look unsuspicious, but has been agreed upon with a trusted third party that this is a cry for help---e.g., uploading a specified text on social media. Even without the adversary knowing the exact text, this scheme does not give enough guarantees; codeword or codephrase that appears to be out-of-place or irrelevant to ongoing events may trigger suspicion~\cite{6890919}. 

In the online case, the adversary may perform statistical tests to detect when a particular communication is sent. Other methods, including hidden information on social media (for example, through hiding data on a text's white space~\cite{SecreTwit,DataHidingSteganography}, or hiding messages in images~\cite{958299}) can still be distinguished by the adversary. If this scheme is to be deployed with the target of reaching as many potential users that may benefit from this, we need to assume that the system we design is public, and the adversary only lacks access to several secret components, also known as Kerckhoff's principle. 

Our approach shares some similarities with \emph{Ask for Ani}, namely that we depend on participation of existing establishments to be the primary receiver of the distress signal, before relaying them. In particular, we believe well-known websites to be well placed to be this intermediary: a wide enough adoption would minimise suspicion based on the recipient of communication, especially as the user can choose these participating websites.

The websites will have to undergo an enrolment process, and this enrolment may be open on a rolling basis. The participation from websites in the scheme is necessary, so we need to ensure that the overhead for the participating sites is as small as possible. With each enrolled website, we utilise the TLS handshake; TLS is a widely deployed system and we minimise the complexity of the scheme deployment on a large scale. In addition, we use public key encryption for the user to communicate with the webserver for scalability. Our strategy is to encode this information into communication objects that are not immediately predictable in the underlying protocol, and we specifically choose the client nonce as it is a value that is chosen on the user side, as well as having a large enough size to include necessary information that can be understood and relayed by the participating websites. Its low predictability means that we can propose a way such that the adversary will not be able to determine whether or not a distress signal is sent with polynomially-bounded statistical tests. However, we note that our scheme can be adjusted to other objects with the relevant properties.

Our method is designed for the user to signal distress against an adversary that has a high degree of control of the network in which the user operates on, as well as having visual access to the user's device screen, which may be a result of being in close physical proximity of the user. Though this would cover most scenarios we have presented, we also cover a wider, stronger set of adversaries who are able to see application layer data of the user's TLS connections---for example, this can be done by setting up a TLS termination proxy.

\subsection{System Model}

\begin{figure}[tp]
\centering
\includegraphics[width=0.8\linewidth]{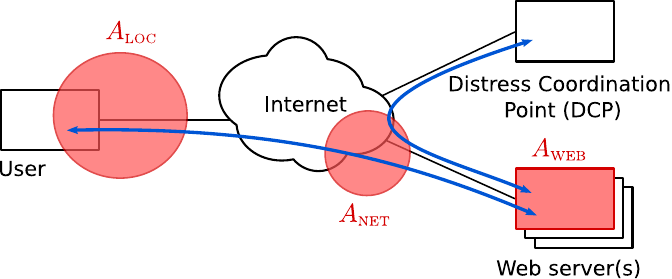}
\caption{System Model. The user communicates with a webserver to signal distress. The webserver in turn forwards the message to \DCP\ which initiates the required actions. $\advloc$ is assumed to be absent during user enrolment.}
\label{figure:systemmodel}
\end{figure}

The three principal players of the system are as follows, and are shown in Figure~\ref{figure:systemmodel}.

\begin{enumerate}
\item The \emph{user}: the player trying to signal distress to a trusted party.  The user's goal is to communicate a distress signal to the trusted party without the adversary's knowledge.
\item The \emph{webserver}: a webserver who relays the user's distress signal to a distress coordination point. To take part in the scheme, the webserver has to go through an enrolment process, and user chooses which webservers to send the distress signal to. We do not, however, assume that the webserver is always honest.
\item The \emph{distress coordination point} (\DCP): a trusted party who receives distress signals. We assume that \DCP\  is always honest.
\end{enumerate}



Any webserver may sign up to participate, though it is up to the user to decide the webserver to send a distress signal to if the choice is plenty (and this amount is up to the system designer, dependent on the use-case). Should there be many options, for the user, this decision can be made primarily based on two things: the trustworthiness of the website, and minimising suspicion from the adversary when visiting the site. 

The trustworthiness of the website can be judged in many ways, for example, the user might choose more popular sites over those with less reputation to lose. To minimise suspicion from the adversary, the user can choose websites that they regularly visit or a popular site they would have a legitimate reason to visit in any circumstance; again, this depends on the specific scenario: during a border check that requires the user to open their email or social media site, these sites may be chosen during enrolment.

In addition, we assume that the user, webserver, and \DCP\ are all able to make a TLS connection and we assume that the user can communicate freely with the webserver at all times. 

Our system operates in three steps:
\begin{enumerate}
\item the \emph{server enrolment}, where the webserver communicates with \DCP\ so the server is added to the \emph{server database}, as well as to facilitate the exchange of necessary encryption keys.
\item the \emph{user enrolment}, where the user communicates with \DCP\ so the user is added to the \emph{user database}, for the user to choose the webservers they trust, as well as the exchange of identifier and encryption keys.
\item the main \emph{distress signal protocol}, which the user uses to signal distress.
\end{enumerate}

Though these steps are separate and somewhat sequential, it is important to note that new webservers and new users can enrol at any time. In addition, webservers and users who have enrolled before are able to update any enrolment details at any time.

\subsection{Adversary Model}

We consider three different adversaries located in different parts of the system, each with different capabilities and goals. Firstly, we have the \emph{network adversary} $\advnet$, an external Dolev-Yao type adversary---that is, we assume that the adversary carries the message and make no further assumptions about confidentiality and integrity of the messages passed on the network. The goals of $\advnet$ is to violate message confidentiality and message integrity (for example, to understand which users are signalling distress, or falsely trigger distress on behalf of an honest user). We also consider a \emph{malicious webserver} $\advweb$ whose goals are similar to $\advnet$ but $\advweb$ acts as a proxy for the users messages and is expected to relay communication to the $\DCP$. 

Lastly we introduce $\advloc$, the \emph{local adversary}. $\advloc$ has control over all of the user's communication channels (so is a subset of $\advnet$, but in addition, has access to the application layer content of TLS and the user's screen. We only assume that the internal state of running applications, including the user's secret keys, can be kept secret, and that $\advloc$ is not present during the user enrolment process and will not have any access to past enrolment communication. This is a new and non-standard adversary model that captures the scenarios we described. The goal of $\advloc$ is to detect a distress signal amongst normal communications. We make this more precise in Section~\ref{section:definition}.


Note that it is possible for $\advloc$ to create their own webserver, i.e. $\advloc$ and $\advweb$ are the same entity. However, as the user chooses which webservers they deem trusworthy, we do not consider this case. 
In addition, $\advloc$ has also the capabilities of $\advweb$ but with additional resources. Henceforth, we consider $\advloc$, $\advnet$, and $\advweb$ as three independent entities. We do not consider correlation attacks with colluding adversaries, for example by traffic or timing analysis, nor do we consider attacks on availability by $\advloc$ or $\advweb$.

\section{Distress Indistinguishability} \label{section:definition}

In this section, we introduce a security definition to capture the situation in which we illustrate how $\advloc$ can use their capabilities in order to detect a distress signal. We first start with the following.

\begin{definition}
For $d \in \{\true, \false\}$, and $x$ from a finite domain, a function $f_N(d,x,k)$ is a fixed-length reversible function if:
\begin{enumerate}
\item $| f_{N}(d, x, k) | = N$
\item Except with negligible probability, there exists $g(y, k')$ such that 
\[
g(f(d,x,k), k') = \begin{cases}
d,x \ \ \text{if } d= \true \\
d \ \ \ \ \text{if } d=\false
\end{cases}.
\]
\end{enumerate}
\end{definition}

The definition above simply states that the scheme preserves a binary value $d$ as well as some additional information $x$ if $d= \true$ that is reversible by using some key $(k, k')$. In addition, such function has a fixed output size~$N$. Of course, the trivial mapping would satisfy the definition, but we also need it to satisfy some security: for our scheme to be secure, we require the function to not only preserve the distress, but \emph{hide} it amongst other connections. Given that the adversary has access to TLS connections, they may have recorded a polynomial amount of client nonces in normal communication when the user isn't intending to signal distress. We formalise this in the following security game.



\begin{definition}\label{def:distressdistinguish}

\noindent Let $f_N$ be a fixed-length reversible function, $x_i$ a finite string, and $pk$ a public key for a public key encryption scheme. Consider the following experiment $\IndistExp(n)$:

\begin{center}
\noindent\resizebox{.8\linewidth}{!}{
	\pseudocode[width=2.5cm]{
	\textbf{Challenger} \< \< \textbf{Adversary} \\[0.1\baselineskip][\hline]
	\< \sendmessageleft*[2cm]{x_1, \hdots, x_n} \< \\
	(pk, sk) \leftarrow \mathcal{K} \< \< \\
	b \leftarrow_R \{0,1\} \< \< \\
	\text{Construct } r_1, \hdots, r_n \\
	\< \sendmessageright*[2cm]{pk,  r_1, \hdots, r_n,} \< \\
	\< \< \text{Guess } b' \in \{0,1\} \\
	\< \sendmessageleft*[2cm]{b'} \< \\
	[0.1\baselineskip][\hline]
	}}
\end{center}

The challenger constructs $r_1, \hdots, r_n$ as follows:
\begin{enumerate}
\item if $b = 0$: $u_i \leftarrow f_N(\false, x_i, pk)$ for each $i=1, \hdots, n$,
\item if $b = 1$: $D$ consists of $n-1$ numbers $r_i \leftarrow f_N(\false, x_i, k)$ and for exactly one $j \in \{1, \hdots, n\}$, $r_j \leftarrow f_N(\true, x_j, pk)$
\end{enumerate}
Define the advantage of $\adv$ on $f_N$ to be:
\[
\IndAdv(n) = | \Pr [\adv \text{ outputs 1} | b = 0] - \Pr [\adv \text{ outputs 1} | b =1] |. 
\]

We say that a fixed-length reversible function $f_N(x, y)$ is $\IndDistress$ secure if for all probabilistic polynomial time adversary $\adv(n)$, there exists a negligible function $\e(n)$ such that ${\IndAdv(n) \leq \e(n).}$
\end{definition}

The experiment above describes the scenario where the adversary is able to receive a distribution of $n$ nonces, and to guess whether or not it contains a distress call, which accurately captures our scenario as the vast majority of TLS connections will be normal communications not signalling distress. When given any amount of nonces, we want to create a scheme where the adversary will not be able to correctly guess whether or not a distress signal exists within that distribution. Intuitively, we need an encryption function to preserve and keep the distress confidential; indeed, note that any fixed-length encryption function would satisfy this, however we do not require the function to return the whole plaintext when $d = \false$. Hence, we use encryption for part of but not all of $f_N$; we will proceed in constructing a fixed-length reversible function that is $\IndDistress$ secure in Section~\ref{section:algorithm}.

\section{Security Protocols}

In this section, we design the architecture where the user can covertly signal distress to \DCP\ via the webserver. Before the distress can be sent, both the user and webserver have to separately go through an enrolment process.

\subsection{Server Enrolment}

The server enrolment aims to include the webserver in the server database. This allows \DCP\ to store the webserver's public key, as well as to agree on a secret key between the webserver and \DCP. This protocol is shown in Figure~\ref{protocol:serverenrolment}.

\begin{figure}[t]
\centering
	\resizebox{.9\linewidth}{!}{\input{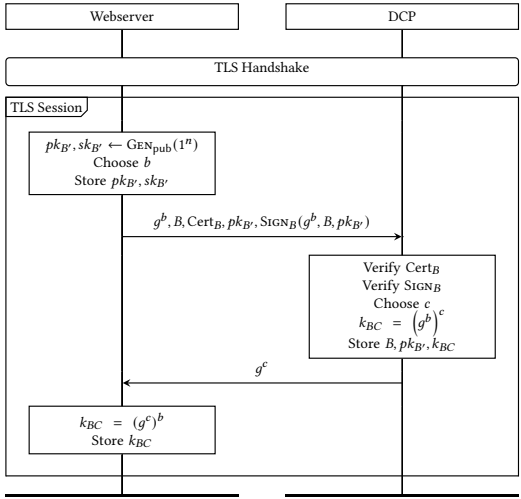}}
	\caption{\textbf{Server Enrolment.} This describes the enrolment process for a webserver to participate in the scheme. At the end of the server enrolment, the \DCP\ has added the webserver to the server database along with the webserver's public key, and the webserver and \DCP\  agreed on a secret, shared key.}
	\label{protocol:serverenrolment}
\end{figure}


The webserver connects to \DCP\ using TLS. The webserver generates a public and private key pair $(pk_{B'}, sk_{B'})$ for the encryption scheme discussed in Section~\ref{section:PKE}, and stores it locally. In addition, the webserver chooses a Diffie-Hellman exponent $b$. The webserver then calculates and sends $g^b$, along with their identity $B$, their certificate $\Cert_B$, their encryption public key $pk_{B'}$, and a signature $\Sign_B(g^b, B, pk_{B'})$ signed using the private key of the certificate $\Cert_B$. Note that $\Cert_B$ is a certificate for the public key of the webserver, not the newly generated encryption keys.

After receiving the message, \DCP\ verifies the certificate $\Cert_B$, and uses the certificate's public key to verify $\Sign_B$; by doing this, \DCP\ verifies that the webserver actually holds the private key of $\Cert_B$, proving authenticity. When all the verification steps have been passed, \DCP\ then chooses their Diffie-Hellman exponent $c$ and computes the Diffie-Hellman key $k_{BC} = (g^b)^c$. 
\DCP\ then stores the webserver's identity $B$, the freshly generated public key $pk_{B'}$, and the shared key $k_{BC}$ in the \emph{server database}, before sending $g^c$ to the server so the server can also compute $k_{BC}$. Here, $k_{BC}$ is the seed for a key derivation function which can be used to derive encryption and $\MAC$ keys used later on in the protocols; we use $k_{BC}$ to describe the keys derived from this seed.


\subsection{User Enrolment}

\begin{figure}[t]
\centering
	\resizebox{.9\linewidth}{!}{\input{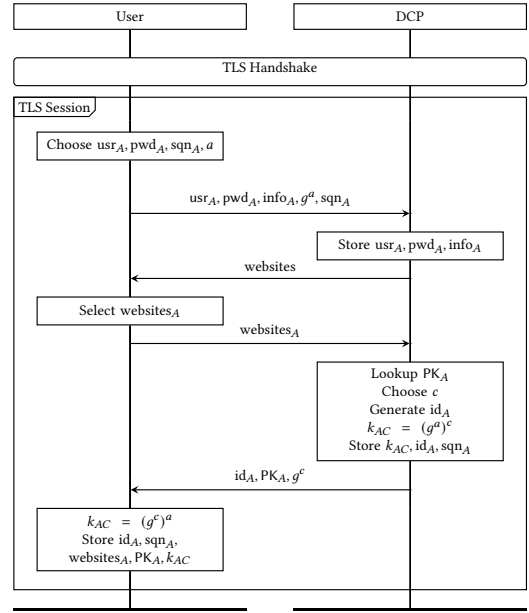}}
	\caption{\textbf{User Enrolment}. This describes the enrolment process for a user to participate in the scheme. At the end of the user enrolment, the \DCP\ has added the user to the user database, the user has obtained the list of webservers they trust along with the public key of each webserver, along with their identifier and sequence number. The user and \DCP\  agree on a shared secret key.}
	\label{protocol:userenrolment}
\end{figure}

The user enrolment process, occurring independently from the server enrolment, allows the user to be added to the user database and for both the user and \DCP\ to agree on necessary shared secrets to which the main protocol rely on: the user identifier, a secret sequence number, and a shared key between the user and \DCP. The user also obtains the list of chosen websites and their public keys. The protocol is shown in Figure~\ref{protocol:userenrolment}.

The user begins a TLS connection with \DCP. The user chooses a username $\usr_A$ and password $\pwd_A$ to register with (or to authenticate themselves if they have previously enrolled), as well as a sequence number $\sqn_A$ and a Diffie Hellman exponent $a$. The user sends $\usr_A$, $\pwd_A$, $\sqn_A$, $g^a$ to \DCP. In addition, the user sends their information $\info_A$, which is needed for \DCP\ to react appropriately in case of distress, as well as a preference on how the user can obtain confirmation, which we discuss in Section~\ref{sec:discussion}.

After receipt, \DCP\ proceeds to store $\usr_A, \pwd_A, \info_A$ in the \emph{user database} and sends back the list of participating webservers from the server database. From the list, the user chooses a subset $\websites_A$, and returns their choice to the \DCP. \DCP\ then chooses a Diffie-Hellman exponent $c$ and computes the shared key ${k_{AC} = (g^a)^c}$. In addition, \DCP\ generates $\id_A$, a unique random string of fixed size $m_{i}$ which acts as the user's identifier in the system, and stores $k_{AC}, \id_A, \sqn_A$ in the user database. Finally \DCP\ sends the user identifier $\id_A$, the list of public keys of websites in $\websites_A$, as well as $g^c$ to the user. This allows the user to compute $k_{AC} = (g^c)^a$, and stores $\id_A$, $\sqn_A$, the list of webservers $\websites_A$ and their public keys $\mathsf{PK}_A$, as well as $k_{AC}$.  Similarly as above, we use $k_{AC}$ as the encryption and MAC keys derived from a key derivation function.


\subsection{Distress Signal Protocol} \label{section:signaldistress}

\begin{figure*}[t]
\centering
	\resizebox{.95\linewidth}{!}{\input{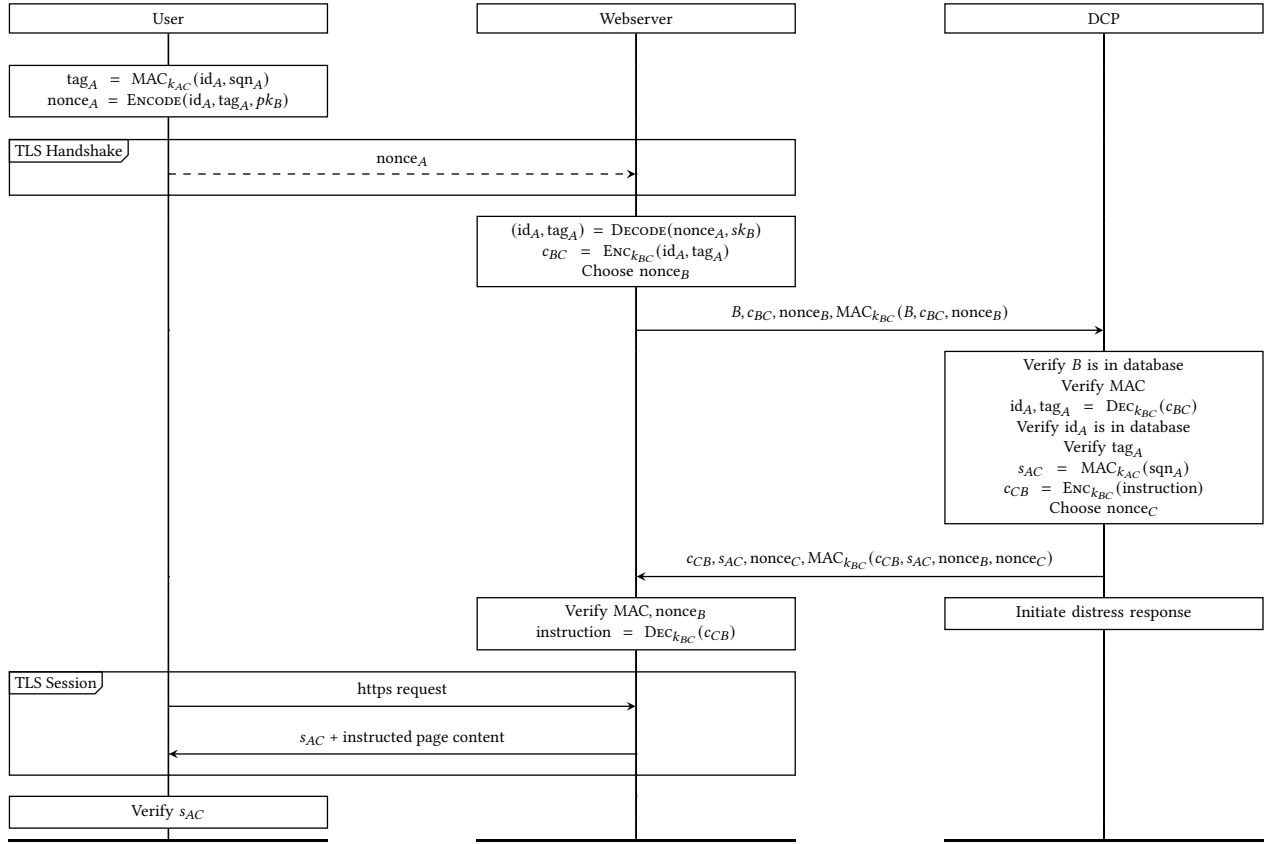}}
	\caption{\textbf{Distress Signal Protocol.} The user, when deciding to send a distress signal, will initiate this protocol and send the encoded distress through the client nonce of a TLS handshake with the webserver. The webserver then relays the user identifier and integrity tag to \DCP\ which will later initiate the agreed distress response. To ensure the user that the distress has been received by \DCP, a verification message is relayed by the webserver to the user.} \label{protocol:main}
\end{figure*}

After the server and user enrolments are complete, the user can now indicate distress at will, including in the presence of $\advloc$. To do so the user must connect to one of the webservers they chose during enrolment. In the following the user connects to webserver $B$ with corresponding encryption key $pk_{B'}$, and the protocol is shown in Figure~\ref{protocol:main}.

The user first computes $\MACtag_A = \MAC_{k_{AC}}(\id_A, \sqn_A)$ using the key shared with \DCP, which requires includes the user identifier and sequence number obtained during user enrolment. Note that when a sequence number is used in the main protocol it is incremented by one, regardless of whether the protocol is later aborted.

The user proceeds to create a \emph{distress nonce} using $\Encode$, a function we will describe in detail in Section~\ref{section:algorithm}. In short, $\Encode$ \emph{hides} the distress, along with $\id_A$ and $\MACtag_A$ within an $N$-bit distress nonce, using $pk_{B'}$. The user initiates a TLS handshake with the webserver and sends the distress nonce as the client nonce in Client Hello, and proceeds to completing the handshake.

Upon receiving a new TLS connection, the webserver unravels the client nonce to check if it has a structure that indicates a distress signal. This process is done by running the $\Decode$ algorithm, described in detail in Section~\ref{section:algorithm}. If $\Decode$ outputs $\mathrm{not \mathunderscore distress}$ then the protocol terminates; otherwise, the webserver obtains $\id_A$ and $\MACtag_A$ as the output of $\Decode$. Termination of the distress protocol does not affect the processing of the main TLS connection, so normal connections function as expected.

If a distress signal is identified, the webserver now encrypts $\id_A$ and $\MACtag_A$ with $k_{BC}$, a shared key that was generated during server enrolment, to obtain $c_{BC}$. The webserver then sends their identity $B$, the ciphertext $c_B$, a nonce $\nonce_B$, as well as $\MAC_{k_{BC}}(B$, $c_{BC}, \nonce_B)$, signed using the shared key $k_{BC}$. 

Upon receipt, \DCP\ checks that $B$ is indeed in the server database and finds the corresponding key $k_{BC}$ so that \DCP\ can verify the $\MAC$. \DCP\ then proceeds to decrypt $c_{BC}$ with $k_{BC}$ to obtain $\id_A, \MACtag_A$. Now, to verify $\MACtag_A$, \DCP\ would use the sequence number $\sqn_A$ agreed in user enrolment to verify $\MACtag_A$. To do this, for a specified $n_{max}$ and for $i = 0, \hdots, n_{max}$, \DCP\ computes $\Vrfy(\MACtag$, $k, (\id, \sqn+i))$ until $\Vrfy$ outputs accepts. Otherwise, the protocol aborts. Note that $n_{max}$ is a parameter that the system designer can specify to allow for malicious or accidental errors in previous distress signals not being received by \DCP. If the MAC verification is successful, \DCP\ updates $\sqn$ to $\sqn'+1$ where $\sqn'$ is the accepted sequence number.

Now, \DCP\ needs to inform the user that their distress has been received, which they do by sending a signed sequence number $s_{AC} = \MAC_{k_{AC}}(\sqn_A)$ to the user. \DCP\ provides instructions to the webserver on how to send $s_{AC}$ to the user, as specified in the user enrolment as part of $\info_A$; this is done by encrypting $\mathrm{instruction}$ with $k_{BC}$ to obtain $c_{CB} = \Enc_{k_{BC}}(\mathrm{instruction})$. \DCP\ then sends the signed $c_{CB}, s_{AC}, \nonce_B, \nonce_C$ to the webserver. After this, the agreed upon distress response may be initiated, as specified in $\info_A$.

Upon receipt, the webserver decrypts $c_{CB}$ to obtain the necessary instructions and verifies $\MAC$ using the shared key $k_{BC}$ using $\nonce_B$ to check that the message is fresh. After verification is successful, the webserver waits for a https request from the user, and sends $s_{AC}$ within the page content as specified in $\mathrm{instruction}$. The user can then verify \DCP's $\MAC$. The specific techniques used to embed the response in the page content can be important if the user is under observation. We discuss ways to do this in practice in Section~\ref{sec:discussion}.

\section{Algorithms}
\label{section:algorithm}

In order for the user to hide the distress, we \emph{encode} the distress into what looks like a TLS client nonce while still satisfying the security requirements of TLS. To do so, we construct a fixed-length reversible function with a specified $N = 256$, the size of the TLS client nonce in TLS 1.3. In addition to hiding the distress state, we are also embedding the user's identifier as well as $\MAC$ tag. The webserver, upon receiving a nonce in the TLS handshake, will have to \emph{decode} what they have received.

Note that should $N$ vary (e.g. if using TLS 1.2), then the bit distributions of the plaintext may be altered accordingly (see Section~\ref{subsection:bitdistribution}). In addition, should $N$ be larger, there may be extra information included (other than distress state, identifier, and tag) such as severity of distress, and further action requested from the user at that instance.

\subsection{Encoding and Decoding Distress}


The main idea of the encoding algorithm is to encrypt a 127-bit integer $r$ with modified Elliptic Curve El Gamal (mECEG, discussed in Section~\ref{section:PKE}), where $r$ contains information about a distress state. If the user is indeed in distress, $r$ would also include their identifier, as well as a $\MAC$ tag for integrity.

In addition, fix $m_d, m_i, m_t$ such that $m_d + m_i + m_t = 127$, with each choice discussed in Section~\ref{subsection:bitdistribution}. Let $\Enc$ be the mECEG encryption algorithm, and $\Encode$, as described in Algorithm~\ref{algorithm:encode} describes how the user encodes their distress and identifier into a nonce.                                                                                                                                                                                                                                                                                                                                                                                                                                                                                                                                                           

\begin{algorithm}[tp]
\caption{Encoding a distress in a 256-bit nonce} \label{algorithm:encode}

\begin{algorithmic}[1]

\Function{\textsc{\Encode}}{$\id, \MACtag,pk$}
	\State \textbf{return} $\Enc_{pk}(1^{m_d} || \id || \MACtag)$
\EndFunction
\end{algorithmic}
\end{algorithm}

The function $\Encode$ is used only when the user wishes to signal distress. It requires three inputs: the user's ID, a $\MAC$ tag, and the public key of the webserver. 
The function simply concatenates a string of $1$s, the user ID, and the $\MAC$ tag. This 127-bit number is then encrypted using $pk$, the public key of the webserver, to produce a 256-bit \emph{distress nonce}.


Algorithm~\ref{algorithm:decode} describes the function $\Decode$, and how the webserver, upon receipt of the nonce, will be able to extract whether or not the user signals a distress; and when they do, also the user identifier and the $\MAC$ tag, the latter used for integrity and freshness to be verified by \DCP.                                                        

\begin{algorithm}[tp]
\caption{Decoding a nonce to obtain distress value}  \label{algorithm:decode}

\begin{algorithmic}[1]
\Function{$\Decode$}{$\nonce, sk$}
\State $r := \Dec_{sk}(\nonce)$ \label{step:decrypt}
\If{$\MSB_{m_d}(r) == 1^{m_d}$}
	\State $\id := \MSB_{m_i}(\LSB_{127-m_d}(r))$
	\State $\mathsf{tag} := \LSB_{m_t}(r)$
	\State \textbf{return} $\true, \id, \MACtag$
\Else
	\State \textbf{return} $\mathrm{not \mathunderscore distress}$
\EndIf
\EndFunction
\end{algorithmic}
\end{algorithm}

We discuss our design choice. As mentioned in Section~\ref{section:signaldistress}, if the user signals distress, $\Encode$ will return a distress nonce, which the webserver will decrypt to obtain $\id$ and $\MACtag$. Indeed, the webserver will have access to $\id$. One may be tempted to encrypt $\id$ so that the webserver has no access to this, however using $\Enc$ takes a lot of space. In addition, a malicious webserver also has access to the user's IP address, so having a the user identifier doesn't add much extra information if the goal of $\advweb$ is to violate privacy.

The user constructs a $\MAC$ tag before inputting it into $\Encode$, which is used to obtain integrity for \DCP. This was added to solve two problems that a scheme without $\MACtag$ would have: (1) the webserver can send distress signals from arbitrary users by guessing $\id$; (2) if the user has previously sent a distress signal, then the webserver is able to replay it and send valid distresses to \DCP\ with the correct $\id$, without the user's knowledge. Using $\MACtag$ which incorporates a sequence number solves these issues.

\sloppy Another option to obtain integrity and freshness, other than using a MAC (and sequence number), is to use a timestamp and simply construct $\MAC_{k_{AC}}(\xtime)$, with \DCP\ checking MAC values of $[\xtime + t, \xtime - t]$ for some buffer period~$t$. However this introduces time synchronisation as an additional requirement. 

Lastly, we did not add an integrity check between the user and the webserver embedded within the client nonce. This is simply because TLS is integrity protected; that is, during the TLS connection between the user and webserver, if $\advloc$ modifies the client nonce, then this will be detected by the client, as the user and webserver will not be able to compute the same handshake secret so the TLS handshake cannot be completed. As we mentioned before, this is an attack on availability that we do not consider.

We now proceed in constructing a fixed-length reversible function which utilises $\Encode$ and $\Decode$.

\begin{prop} \label{prop:distresshidingfunction}
Let $(pk, sk)$ be a key pair for a public key encryption $\Enc$. The following function 
\[
f_N \left(d, (\id, \MACtag), pk\right) = 
	\begin{cases}
	\Encode( \id, \MACtag, pk) &\text{if } d = \true \\
	\PRG(N) \ &\text{if } d = \false
       \end{cases} 
\]
where $\PRG(N)$ outputs a pseudorandom string of size $N$, is a fixed-length reversible function. 
\end{prop}

\begin{proof}
Let $g_N(y, sk) = \Decode(y, sk)$. Given how we choose $\Enc$ in $\Encode$ to have a fixed output size (see Section~\ref{section:PKE}), then the first criteria is satisfied.

Let $d = \false$. Then $g_N(f_N(d, (\id, \MACtag), pk), sk) = \Dec_{sk}(r)$ where $r$ is an output of $\PRG(N)$. Then $\Decode(r, sk)$ will output $\mathrm{not \mathunderscore distress}$ for all $r' = \Dec_{sk}(r)$ except for $r'$ such that $\MSB_{m_d}(r') = 1$, and the probability of that happening is $1/2^{m_d}$ as $r$ is pseudorandom. Hence, if $m_d$ is large enough, then $d$ is preserved except with negligible probability.

Now consider the case when $d = \true$. We have $\Encode(\id, \MACtag, pk) = \Enc_{pk}(r)$ where $r=1^{m_d}||\id || \MACtag$. Now, given that $(pk, sk)$ is a public key encryption pair, then $\Dec_{sk}(\Enc_{pk}(r)) = r$, except for negligible probability of decryption error. Now, as $\MSB_{m_d}(r) = 1^{m_d}$,  $\Decode$ returns $\MSB_{m_i}(\LSB_{127-m_d}(r)) = \id$ and $\LSB_{m_t}(r) = \MACtag$ by construction. That is, indeed, $\Decode(\Encode(\id, \MACtag, pk), sk) = \true, \id, \MACtag$.
\end{proof}

Note that the webserver, upon receipt of a TLS handshake from \emph{any} user, would decrypt the client nonce as part of $\Decode$, regardless of whether or not the user is enrolled in the scheme. Due to our choice of $(\Gen, \Enc, \Dec)$, given the original security properties of a nonce, a false positive will happen with probability $\frac{1}{2^{m_d}}$. If such a false positive occurs, the webserver will also obtain some $\id'$ and $\MACtag'$ and forwards them to \DCP. That is, other than receiving a distress signal directly from a user, the webserver also has the function of filtering out false positives. 

Upon receipt of $\id', \MACtag'$ from a user that did not signal distress, \DCP\ will check this with the user database. The probability that the $\id$ is valid is $\frac{N_{user}}{2^{m_i}}$, where $N_{user}$ is the total number of enrolled users. If, already with probability $\frac{1}{2^{m_d}} \cdot \frac{N_{user}}{2^{m_i}}$ such $\id'$ is valid, then the probability that a valid ${\MACtag' = \MAC_{k}(\id', \sqn)}$ occurs is negligible.


\subsection{Bit Distribution} \label{subsection:bitdistribution}

\begin{figure}[tp]
  \centering
  \includegraphics[width=\linewidth]{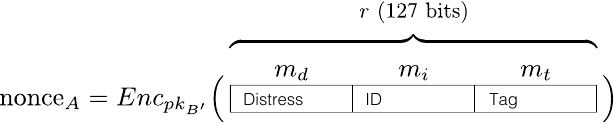}
  \caption{\textbf{Encoding of nonce data.} The input to the public key encryption is a concatenation of a distress signal, an identifier, and an authentication tag with lengths of $m_d$, $m_i$, and $m_t$ bits respectively. The three lengths have to add up to 127 bits which will yield a 256 bit output of mECEG encryption.}
  \label{fig:bitdist}
\end{figure}

The output of $\Encode$ needs to contain three different pieces of information: a marker that lets the webserver determine whether this is a valid distress signal or a normal nonce, an identifier, and a tag that authenticates the user to the \DCP. In this section we discuss the best way to allocate the limited number of available bits to these three fields. We denote the length of the three fields as $m_d$, $m_i$, and $m_t$ for number of bits for distress, identifier, and tag, respectively. This is visualised in Figure~\ref{fig:bitdist}. Note that we have a limited number of space for $m_d$, $m_i$, and $m_t$ given that $m_d + m_i + m_t = 127$ in order to fit in to the size requirement of client nonce.


The webserver always runs $\Decode$ on every TLS connections initiated. When $\Decode$ is called on a received nonce which shows distress (i.e. $\MSB_{m_d}(r) = 1^{m_d}$), there are three possibilities: (1) the user with corresponding $\id$ has signalled distress, (2) a false positive from a user enrolled in the system, (3) a false positive from a user not enrolled in the system. 

The probability of a false positive happening, from \emph{any} user, is~$1/2^{m_d}$. Hence, we want $m_d$ to be as large as possible as this would reduce the probability of false positives. Ideally, $m_d \geq 17$ so that less than~1 in~100,000 TLS connections would come up as false positives. For $m_d = 20$, this goes up exponentially to~1 in~1,000,000.

Now, we want $m_i$ to be large enough so that all users who want to enrol are able to, and that the scheme is more resistant to resource depletion attacks. In addition, with a larger $m_i$, a malicious webserver $\advweb$ is less likely to be able to `guess' arbitrary $\id$s, though even though they succeed in guessing a valid $\id$, the probability that they will pass the integrity verification through $\MAC_{k_{AC}}(\id, \sqn)$ should be negligible---hence, lastly, $m_t$ indicates the size of the $\MACtag$, and the larger $m_t$ is, the better $\MAC$ security obtained. 

Though we do not specify the exact size of $m_d, m_i,$ and $m_t$, we illustrate the case for $m_d = 32$. That is, the probability of the webserver in receiving a false positive is $1/2^{32}$, which, on average, means that the webserver receives one false positive in every 4,294,967,296 TLS connections. This leaves $m_i + m_t = 95$, so if $m_i = 32$ and $m_t = 63$, the user identifier gives space to over~4 billion users, and a $\MAC$ tag with output size of~63-bits. 


\subsection{Elliptic Curve El Gamal} \label{section:PKE}


A background of elliptic curves and elliptic curve El Gamal (ECEG) is laid out in Appendix~\ref{appendix:EC}.

Traditionally, a ciphertext is of the form $(P,Q)$ where ${P = (X_P, Y_P)}$, ${Q = (X_Q, Y_Q)}$. Note that for an elliptic curve ${Y^2 = X^3 + AX + B}$ over $\F_q$, then we use the modern approach of simply sending the $x$-coordinate of $P$ and $Q$, as well as two additional bits $b_P$ and $b_Q$ showing which square root of $X^3 + AX + B$ are the $y$-coordinates of $P$ and $Q$ respectively. This is an inexpensive computation by the receiver, and allows us to save space~\cite{menezes}. We use this approach, which we refer to as \emph{modified} El-Gamal encryption (mECEG) in throughout this paper.

Unfortunately, the ciphertexts produced from mECEG are not indistinguishable from random. Indeed, because ciphertexts are elliptic curve points, they follow some structure that allows an adversary to come up with a strategy when trying to see if a ciphertext is an output of an mECEG encryption, instead of chosen uniformly from random. There are several attacks that an adversary can do, as discussed in~\cite{10.1145/2508859.2516734}, but the most severe one is simply checking whether ${X^3 + AX + B}$ is a quadratic residue in the field, which the adversary may succeed with probability $\frac{1}{2}$. We show later, that even with this extra information, the adversary will not be able to distinguish a distribution which contains a distress signal against the uniform distribution.

We chose mECEG instead of other encryption schemes as it satisfies the two main conditions: a ciphertext size as a function of the plaintext size, as well as its semantic security, so that TLS security is maintained. Other encryption schemes, for example \randomCPA\ encryption schemes would have a better property: their ciphertexts are indistinguishable from random. However, \randomCPA\ schemes are not practical in this sense: the ciphertexts are too long to fit into the 256-bit space (without compromising on security), and the decryption takes 128 exponentiations, which hinders scalability as webservers will unlikely wish to have such overhead.
\section{Computational Overhead} \label{section:implementation}

As the participation of webservers is necessary for the scheme to function, it is important that our design does not generate a large amount of overhead in computation and communication for the website.

Should a webserver participate in our scheme, whenever a TLS connection is requested, the webserver is required to, in parallel to completing the TLS handshake, decode the client nonce received during Client Hello. Note that for the webserver,  a TLS handshake already contains two expensive operations: the calculation for the early shared secret key by multiplication of elliptic curve points, and signing a hash using the certificate's public key for authentication. Our protocol adds an additional computational overhead of one elliptic curve multiplication during decryption. The additional calculation of two $y$-coordinates of the ciphertext elliptic curve points are inexpensive~\cite{menezes}.

After decryption is completed, the webserver checks whether the plaintext indicates a distress, through a bitwise AND operation, which is computationally cheap. Should the plaintext reveal a distress, the webserver proceeds to encrypt the user ID and tag and compute a MAC, with symmetric keys obtained during enrolment. After \DCP\ verifies that the distress signal is from a valid user, then the webserver verifies a MAC and decrypts the message from \DCP\, again symmetrically. These operations are not computationally expensive. That is, design choices we made have been to minimise this overhead to one extra elliptic curve multiplication operation, and an increase in distress received only affects overhead linearly.

We have previously argued that the number of false positives received by the webserver is low for $m_d = 32$. To illustrate this, consider Google search, which is the most visited website in the world by traffic~\cite{semrush}. Google on average receives over 40,000 search queries per second~\cite{GoogleStat}, or over 3.5 billion per day. Although there is no official data on how many servers are in Google data centres (a 2016 estimate is 2.5 million servers~\cite{GoogleServers}), if 100,000 servers are dedicated to Google Search and assuming that the TLS connections are distributed evenly across all servers, this means that for each server, a distress false positive happens once every 34 years. Even in the worst-case-scenario where all the connections go to a single server, on average, a distress false positive will happen every 29.83 hours. 

\section{Security Analysis} \label{section:security}

We assume the following:

\begin{enumerate}
\item[A1] The certificate authority issuing certificates of each protocol party is honest. \label{assumption:CA}
\item[A2] The signature scheme used by each protocol party is secure; that is, the probability of successfully forging a signature without access to the secret key is negligible. \label{assumption:signature}
\item[A3] The MAC used by each protocol party is secure. \label{assumption:MAC}
\item[A4] The probability that an honest party picks the same nonce twice is negligible. \label{assumption:nonce}
\end{enumerate}

\subsection{Server and User Enrolments}

We note down the functional and security guarantees from the server and user enrolments below. 

\begin{guarantee}[Server Authenticity] \label{guarantee:serverauthenticity}
The server enrolment protocol will only complete successfully, if the webserver is able to produce a valid signature corresponding to its certificate (real world identity). 
\end{guarantee}

\begin{proof}
\sloppy
If an adversary wants to claim that they are the webserver $B$, they have to successfully send $x$, $B$, $\Cert_B, pk_{B'}$, $\Sign_B(x, B, pk_{B'})$. 
To replay the message, they need to have previously captured $\Sign_B(x, B, pk_{B'})$ from another enrollment, but each enrollment is done in its own TLS session so this message is not available to the attacker. To craft the message, the attacker has to produce $\Sign_B(x, B, pk_{B'})$ using a private key associated with $\Cert_B$. Given our Certificate Authority is a trusted entity (A1) and the underlying signature scheme is secure (A2) this is not possible. 
\end{proof}

Functionality-wise, the above guarantee means that \DCP\ can be convinced that the identity calimed by the webserver indeed belongs to the webserver.

\begin{guarantee}[Server Shared Key] \label{guarantee:kBC}
  If the server enrolment protocol is completed successfully, then the webserver obtains a symmetric key that is known only to the webserver and \DCP.
\end{guarantee}

\begin{proof}
Firstly, \DCP\ is authenticated during the TLS handshake, and the webserver is authenticated by Guarantee~\ref{guarantee:serverauthenticity}. Both $g^b$ and $g^c$ are sent within the same TLS session, so confidentiality and integrity follows.
\end{proof}

\begin{guarantee}[DCP Authentication]
  If the user enrolment protocol is completed successfully, then the user can be certain that they have been speaking to \DCP.
\end{guarantee}

\begin{proof}
  \DCP\ acts as the TLS server in the TLS connection, and is thus authenticated during the TLS handshake.
\end{proof}

\begin{guarantee}[User Shared Key and Registration] \label{guarantee:kAC}
  If the user enrolment protocol is completed successfully, then the user is registered and the user obtains a symmetric key and sequence number that is known only to the user and \DCP.
\end{guarantee}

\begin{proof}

The user has a TLS session with DCP by Guarantee \ref{guarantee:serverauthenticity}. Both $g^b$ and $g^c$ are sent within the TLS session, so confidentiality and integrity of the key follows.
\end{proof}


\subsection{Main Protocol}

We show the security guarantees from our main protocol, when the user signals distress. A crucial security property we want is captured in the following theorem.

\begin{prop} \label{thm:indistinguishability} Let $\Pi = (\Gen, \Enc, \Dec)$ be modified Elliptic Curve El Gamal encryption scheme used by the function $\Encode$. If $PRG$ is a secure pseudorandom number generator then the following fixed-length reversible function
\begin{equation*}
f_N\left(d, (\id, \MACtag), pk\right) = 
	\begin{cases}
	\Encode( \id, \MACtag, pk) &\text{if } d = \true \\
	\PRG(N) \ &\text{if } d = \false
       \end{cases} 
\end{equation*}
with $g_N(r, sk)= \Decode(r, \sk)$ is $\IndDistress$ secure.
\end{prop}
\begin{proof}
The strategy of this proof is simple: statistically close distributions are polynomially indistinguishable~\cite{GOLDWASSER1984270}. Let $D_b$ be the distribution of nonces formed when bit $b$ is chosen, i.e. if $b= 0$ then $D_b$ is simply the uniform distribution, and when $b=1$ then there exists precisely one distress signal in $D_1$. Let $U$ be the uniform distribution (so the range of $D_0$ and $D_1$ are $U$) and $V$ be the distribution of elliptic curve points. The statistical distance between $D_0 = \{X_1, X_2, \hdots\}$ and $D_1 = \{Y_1, Y_2, \hdots\}$ is:
\begin{align*}
\Delta(D_0, D_1) &= \sum_u | \Pr(X_i = u) - \Pr (Y_i = u) | \\
\end{align*}
Now, for all $u \in U$ except for one distress point $d$, $\Pr(X_i = u) = \Pr (Y_i = u)$. Also note that the number of possible outputs of El Gamal is~$\frac{2^N}{4}$ as there exists two $X$-values in the ciphertexts, each lying on the curve with probability~$\frac{1}{2}$. Hence,
\begin{align*}
\Delta(D_0, D_1) & =  | \Pr(X_i = d) - \Pr (Y_i = d) | = \frac{1}{2^N} - \frac{1}{n} \cdot \frac{4}{2^N} = \frac{n-4}{n \cdot 2^N},
\end{align*}
which is negligible.
\end{proof}

To demonstrate that this is the case, a possible adversary strategy that uses elliptic curve structure is provided in Appendix~\ref{appendix:adversaryadvantage}.

\begin{guarantee} \label{guarantee:bobauthentication}
If the main protocol is completed successfully, then for \DCP, it is guaranteed that the message from the webserver is fresh and indeed comes from an enrolled webserver.
\end{guarantee}

\begin{proof}
To break the guarantee, an adversary needs to send $B, c_B, \nonce_B, \MAC_{k_{BC}}(B, c_B, \nonce_B)$ for an enrolled webserver $B$. The adversary has to send a valid $\MAC$ signed using $k_{BC}$, a shared secret key between the webserver and \DCP. This key is known only to the webserver and \DCP, by Guarantee~\ref{guarantee:kBC}, and in particular the adversary doesn't have access to it. The adversary has two options in sending $\MAC_{k_{BC}}(B, c_B)$: by forgery or replay. Forgery is impossible by Assumption~(A3). If the adversary replays a message, $c_B$ would be used twice, but $c_B$ contains the sequence number so it cannot be reused. 
\end{proof}

\begin{guarantee}\label{guarantee:charlieauthentication}
If the main protocol is completed successfully, it is guaranteed that the reply from \DCP\ is a response to the recent communication from the webserver, and indeed comes from \DCP.
\end{guarantee}

\begin{proof}

Firstly, assume that the $\DCP$ has never received the request from the webserver. That means that the $\DCP$ does not have access to $\nonce_B$, and hence they will be unable to construct $\MAC_{k_{BC}}(...,\nonce_B, ...)$ and the $\MAC$ verification by the webserver will fail.

Now for an adversary to claim to the webserver that they are \DCP, they need to send $m_1, m_2, \nonce_C$,$ \MAC_{k_{BC}}(m_1, m_2, \nonce_B, \nonce_C)$ successfully for an arbitrary $m_1, m_2, \nonce_C$. The adversary can do this by creating their own message or replaying a previously captured one. By Guarantee~\ref{guarantee:kBC}, only the webserver and \DCP\ know the key $k_{BC}$, so the adversary can only successfully forge the $\MAC$ with negligible probability by Assumption~(A3). To replay the message from \DCP\ the adversary would have to capture a previous message using the same $nonce_B$, however by Assumption~(A4) this happens only with negligible probability.

\end{proof}

\begin{guarantee} \label{guarantee:charliereceive}
If \DCP\ receives a valid $\id_A$ and a valid $\MACtag_A$, then it is guaranteed that the user with corresponding $\id_A$ has recently signalled distress.
\end{guarantee}

\begin{proof}
By Guarantees~\ref{guarantee:bobauthentication} and~\ref{guarantee:charlieauthentication} we know that the webserver and \DCP\ are indeed speaking to one another, so an adversary cannot be of the form $\advnet$ or $\advloc$. Hence, an adversary has to be $\advweb$.

For $\advweb$ to break this guarantee, they have to send $B$, $c_B$, $\nonce_B$, $\MAC_{k_{BC}}(B, c_B, \nonce_B)$ successfully without the user ever signalling distress, or by replaying a previous distress signal sent by the user.

We consider the first case. For any $\nonce_A$ that the adversary receives from the user, $\Decode(\nonce_A, sk_B)$ will output $\mathrm{not\mathunderscore distress}$ (except with negligible probability) and in particular, the adversary does not have access to $\id_A$. Hence, the adversary will need to guess $\id \leftarrow_R \{0,1\}^{m_{i}}$ and sends through $B, c_B, \nonce_B, \MAC_{k_{BC}}(B, c_B, \nonce_B)$ with ${c_B = \Enc_{k_{BC}}(\id, \MACtag_A)}$. Upon verification, $\advweb$ passes the identity verification since $B$ is in the server database, and $B$ has correctly signed a $\MAC$ with the shared key. However, \DCP\ upon obtaining $\id$ from decrypting $c_B$, will only pass the verification with probability $\frac{N_{user}}{2^{m_i}}$ where $N_{user}$ is the total number of users registered in the scheme. 
In addition, $\advweb$ needs to forge $\MACtag_A = \MAC_{k_{AC}}(\id, \sqn_A)$, which succeeds with negligible probability according to Assumption~(A4).

We now consider the second case, where the user has signalled distress in the past. In this case, $\advweb$ will receive $\id_A, \MACtag_A$ and as above, sends this in encrypted form to \DCP. Now, \DCP\ will receive a valid $\id_A$, but $\advweb$ has to pass the $\MAC$ verification: that is, $\advweb$ needs to send $\MACtag_A = \MAC_{k_{AC}}(\id_A, \sqn_A)$ successfully. According to Guarantee~\ref{guarantee:kAC}, both $k_{AC}$ and $\sqn_A$ are shared only between the user and \DCP. $\advweb$ can try to forge a $\MAC$ without the $k_{AC}$, which is impossible except with negligible probability (Assumption~(A4)), or replay a previously sent $\MACtag_A$. However, this requires the user to send the same $\sqn_A$ more than once, which is not possible by design.
\end{proof}

We proceed to the main guarantee for the user:

\begin{guarantee} \label{guarantee:main}
If the user signals distress and the protocol completes successfully, then the user is certain that \DCP\ has recently received the distress.
\end{guarantee}

\begin{proof}
For the adversary to convince the user that \DCP\ has received the distress, the adversary needs to send $\MAC_{k_{AC}}(\sqn_A)$ to the user and pass the verification. Firstly, given that the this communication happens through a TLS session, it is secure against $\advnet$ and $\advloc$. Hence, an adversary has to be $\advweb$.

To send a correct $s_{AC} = \MAC_{k_{AC}}(\sqn_A)$, $\advweb$ can do this in two ways: forge the MAC or replay a previous message. By Guarantee~\ref{guarantee:kAC} the key $k_{AC}$ is known only to the user and \DCP, so any attempt at forging $\MAC_{k_{AC}}(\sqn_A)$ will only succeed with negligible probability according to Assumption~(A3). To replay a previous message $\MAC_{k_{AC}}(\sqn_A)$ will fail on the user's verification, as $\sqn$ will be outdated.

We conclude that the message $\MAC_{k_{AC}}(\sqn_A)$ has been sent by \DCP\ is fresh and without alteration from $\advweb$. \DCP, being an honest party, sends this message only after receiving a valid connection from the webserver; by Guarantee~\ref{guarantee:charliereceive}, the user has indeed recently signalled distress. 
\end{proof}

Note that there are a few situations in which $\advweb$ are able to  act maliciously within our system. Firstly $\advweb$ in receipt of a distress signal from the user, may not forward that signal to \DCP; however $\advweb$ does this, the user will not receive the correct feedback, and using Guarantee \ref{guarantee:main}, the user knows that the distress signal is never received by \DCP. Secondly, if the user sends a distress signal which is forwarded from $\advweb$ to \DCP, after \DCP's receipt of the distress $\advweb$ may not send the confirmation message to the user, which may lead to the user believing that the distress has never been received by \DCP. These are attacks on availability and are not solvable by any protocol, and therefore out of scope. In practice, if the user does not receive a confirmation from a particular website, they can simply try signalling distress to another webserver.

\section{Related Work} \label{section:litreview}

This section reviews related work. We focus on three separate topics: related adversary models, work dedicated to designing a method for a distress signal to be sent, and lastly, undetectable communications.

\textbf{Adversary Models.} Our adversary model is related to those discussed under \emph{endpoint compromise}, which the adversary would be considered to have full access to device secrets. This is indeed the model described under \emph{insider threat}~\cite{10.1007/3-540-36415-3_2,10.1145/2904018} for an adversary who has legitimate, full access to a device and wishes to exfiltrate data through a monitored network, usually a covert channel. In~\cite{DECIM} the authors considers an adversary who has full access to the network and messaging server, and is able to periodically compromise a device but does not consider the case should the adversary is aware that the user knows, and acts on, the compromise (trivially through removing the adversary's access).

Given how our adversary model describes a high degree of control of the network, this might seem similar to a censorship resistance system (well reviewed in~\cite{7546542,khattak_sok_2016}), where a user wishes to exchange communication with a receiver through a communication channel controlled by a censor who actively attempts to prevent this communication. Censors may have a set of distinguishers which takes in network traffic to be analysed and outputs `accepted' traffic flows. This is similar to our scenario where the adversary wishes to distinguish a distress communication from a normal communication on the network, as we assume communications may take place. However, in our case, the adversary has the additional capability due to their locality to the user.\\

\textbf{Distress Signalling.} Coercion attack, sometimes called by the term \emph{rubberhose cryptanalysis} is a physical attack to bypass security, usually authentication, by physically forcing the user to reveal their secrets. A practical solution to a coercion attack is the use of a second authentication secret as part of a covert mechanism, which would alert distress during an authentication process (\cite{davida_anonymity_1997,clark_panic_nodate,10.1145/1866886.1866895,funkspiel,Zhao2015GracewipeSA} and in commercial products~\cite{safewatch}). Though these schemes allow the adversary to have physical and logical access to the devices, the adversary does not have access to the network and only has access to the private information if authentication is granted. \\

\textbf{Undetectable Communications.} The initial study of undetectable communications was encapsulated in the prisoner's problem~\cite{Simmons1984}, which describes two prisoners trying to communicate to each other through a warden, who is able to modify the messages and allow exchange of messages only if they look innocuous.

The scenario of an insider threat wanting to covertly send information as well as someone circumventing censorship are usually achieved by a scheme which ensures stealth. There are numerous ways in which this can be done, through steganographic methods or covert channels ~\cite{10.1145/2897845.2897913,HTTPcovert,NTPCovert,IPv6Covert,maurice2017hello}.

The notion of \emph{undetectable communication} specifically for the case of Online Social Networks (OSNs) was explored in~\cite{6890919}, where formal definitions are introduced. In~\cite{10.1145/2435349.2435351}, the authors proposed a system to allow users to exchange data over existing web-based sharing platform while maintaining confidentiality, as well as hiding from a casual observer that the exchange is taking place.

For steganographic or censorship resistance purposes, there are several work looking into encryption functions that have pseudorandom ciphertexts, both in the private-key setting~\cite{10.1007/3-540-45708-9_6} and public-key setting~\cite{von2004public}. In~\cite{10.1145/2508859.2516734}, building on~\cite{10.1007/978-3-540-30108-0_21}, the authors proposed a method to encode elliptic curve points indistinguishable from uniform random strings, which would assist elliptic curve cryptography to be suitable as a censorship circumvention tool. Unfortunately, these constructions have the downside of having extremely long ciphertexts. The constructions in~\cite{10.1007/978-3-540-30108-0_21,10.1007/978-3-642-15317-4_18} also produces pseudorandom ciphertexts, while aiming for space efficiency. 
Our approach is similar to that in~\cite{10.1007/10958513_13,266599}, where both papers proposed methods of embedding information in the TLS client nonce. The former was proposed to embed an Escrow Agency Field (EAF) in protocol fields, with the key difference with our scheme is that encryption happens at \emph{every} communication, not only during specific instances hence why their security reduces directly to semantic security of ECEG. In the latter, the client assumes a symmetric shared key with the \emph{decoy router} as opposed to using public keys.

\section{Discussion}
\label{sec:discussion}

In this section we discuss a few topics with additional points we wish to clarify.

\subsection{Enrolment}

To make use of our proposed scheme, a user is required to enrol when the adversary is not present. This is important to allow the main distress signal protocol to be initiated within a single action and for the security guarantees be met. We stress that our scheme is not suitable for situations where enrolment is not possible.

Acts of abuse are often difficult to recognise, and our scheme is only appropriate to be set up in cases where the user has recognised that they are in a vulnerable position and are able to prepare should they decide to seek help. Unfortunately, in many cases domestic abuse survivors unable to leave their abuser, or return to the abuser after an initial offence~\cite{safelives1,safelives2}.

We emphasise that we do not intend to minimise the difficulty of helping domestic abuse survivors. Our proposal can fit into a set of solutions that are currently available---from practical advice in managing personal devices and communications \cite{womensaid,applemanual}, to the support given by organisations supporting survivors of domestic abuse, and can be a valuable building block for future work.

This enrolment procedure is similar to those used in practice for other related schemes, for example Path Community~\cite{pathcommunity} were created for individuals to walk home safely, where one might add emergency contact numbers or guardians when they set up the application. When enrolling in this scheme, some personal information may be included so the backend can act appropriately, for example by contacting an emergency contact or forwarding the case to the relevant authority.

In other situations, for example journalists or activists travelling through a hostile nation, preparations can easily be made to face a situation where a distress signal may be required. In this case, the employer or organisation of the individual can set up an architecture where they are the  \DCP. The response to a distress by the individual may include contacting the relevant embassy including the necessary personal details, including location. Our scheme does not specify the particular user information required for enrolment, nor the response from the backend, as it is incredibly context-specific and we do not wish to narrow down these many possibilities of implementation.

\subsection{User Interface}

Our scheme specifies what happens in the background when distress is initiated by the user's device, but we have not specified the user's interactions with the device when initiating the distress signal, or in the verification of response from the \DCP\ through the webserver.

Firstly, the enrolment phase includes an installation of the distress system on the user's device. When the adversary is present and the user wishes to initiate the signal, some action is required from the user to do so. 

There are plenty of ways in which such action can be done stealthily, and there are plenty of existing techniques which can be implemented. For example, the user can initiate it by pressing a specific keyboard shortcut as they enter the website address (see, for example, \cite{shortkeys}), or entering a duress password if the interface is designed in a way that doesn't reveal any information about what's being entered, or pressing a hidden icon on the browser. Note that an action has to be designed to be done \emph{before} the user enters the website as the distress is sent during the TLS handshake, and not the session. There are many considerations a system designer may consider within this UI challenge, including stealth and usability, while reducing false positives through accidental inputs. At the same time, this action needs to be simple and easy to remember for the user who might be in a state of duress when sending the distress signal. This is an interesting challenge and a considerably important one, which we consider as future work on the deployment of this architecture, while noting that this design is entirely context-dependent.

In our protocol, the user verifies the response $s_{AC}$ to confirm that \DCP\ has indeed received the user's distress signal. The $s_{AC}$ is sent through the TLS connection from the webserver as part of the TLS session---that is, within the page content.  The $\mathrm{instruction}$ for how best to embed the confirmation in a page can be selected by the user at enrolment and communicated to the webserver by \DCP\ when a distress signal is received. To not limit the possibilities of implementation, we have intentionally not specified how, but we give several options below.

With a traditional adversary model, something as simple as including $s_{AC}$ as a HTTP-header that a plugin can verify would work. However, with $\advloc$ this may be problematic if the user is under strict surveillance, as such a header (or any other overt confirmation) would be visible to the adversary---this would defeat the user's goal of stealthiness. Instead, a solution could be that the confirmation is embedded as a specific content that will be unnoticed by the adversary. 

An ideal embedding would be one that gives the user positive confirmation that the distress signal has been received, and at the same time allows for $s_{AC}$ to be extracted by a plugin (or the browser) to be verified cryptographically. Any image generated from a seed (e.g., randomart~\cite{randomart}), can be easily verified, but whether that would appear ``normal'' on a website, is very context dependent. For websites that commonly include advertisements, it may be possible to display selected advertisements as confirmation, where the visual details (or indeed the text) encode the exact value of $s_{AC}$. If a website commonly includes dynamic content like video, a response that includes moving images or sound, are options. Particularly with video, there is the opportunity to include other measurable, non-visible parameters such as time delay and use timing-based covert channel techniques. We refer to existing research on undetectable communication~\cite{6890919,10.1145/2435349.2435351}, for additional examples.

In short, the response sent between the webserver and the user needs to be: (1) invisible to the adversary on the TLS connection, e.g., it is important that $s_{AC}$ is not being sent on its own, but is embedded within a page content so the existence of $s_{AC}$ is obfuscated, (2) verifiable by the user's device, (3) visually verifiable by the user, and (4) unnoticeable to the adversary who has visual access of the device. These are important design choices that need to be made when the architecture is deployed, and decisions are to be made dependent on the context and situation of the user as well as the webserver's purpose and design. We note that there is no one-size fits all solution that covers all of the possible use-cases of this architecture.


\section{Conclusion}

We introduced a new adversary model $\advloc$, who not only has visual access to the user's screen, but full control over all of the user's communication channels as well as the application layer content of the TLS. This includes scenarios including a domestic abuser, a compromised journalist, or unlawful arrest and search. We introduced distress signalling as a goal for a user to perform against such adversaries. We formalised the adversary's abilities in the form of a security game,  and the adversary's goal is to detect a distress signal within a set of normal communications.

To fulfil the user's goal, we constructed an architecture where the user can signal distress using webservers as an intermediary. To do so, the distress is embedded along with additional information as the client nonce in the TLS handshake between the user and the webserver; we specified this in our \emph{encode} function, which uses the Elliptic Curve El Gamal public key encryption scheme. Upon receipt, the webserver \emph{decodes} the distress and forwards the extra information to the backend so that appropriate action can be taken.

Using TLS as well as public key encryption allows for better scalability as well as stealth, and the protocol can coexist with normal uses of TLS. For the webserver, the decoding function only contributes to computational overhead equivalent to one exponentiation per TLS request, and a false positive happens very rarely, so there is very little overhead in terms of communication. When false positives on the webserver's side happen and forwarded to the backend, the backend will perform further integrity checks, so that actual false positives happen only with negligible probability. 

We performed a full security analysis of our architecture, including proving distress indistinguishability of our scheme. Lastly, we discussed further practical considerations that a full scale deployment needs to address.

\bibliographystyle{ACM-Reference-Format}
\bibliography{references}
\appendix
\section{Background on Elliptic Curves} \label{appendix:EC}

The notation in this section is standard, and further details can be found on~\cite{silverman}.

Let $\F_q$ be a finite field of size $q$, with $q$ a large prime power. An elliptic curve $\cale$ over $\F_q$ can be written in the form ${\cale_{A,B} = y^2 + x^3 + Ax + B}$, with $A,B \in \F_q$, with $B^2-4AC \neq 0$ and the point at infinity $\mathbf{\underline{o}} = [0:1:0]$. 

We denote $\cale(\F_q)$ the set of $\F_q$-rational points of $\cale$, that is, for $\cale = \cale_{A,B}$,
\[
\cale(\F_q) = \{(x,y) | y^2 = x^3 + Ax + B\} \cup \{\mathbf{\underline{o}}\}.
\]

An important theorem concerning the number of points of an elliptic curve is the following.

\begin{theorem}[Hasse]
Let $\cale$ be an elliptic curve defined over $q$. Then 
\[
|\#\cale(\F_q) - q - 1 | \leq 2\sqrt{q}.
\]
\end{theorem}

Note that if $x^2 = a$ in $\F_q$, then $(-x)^2 = a$. Combined with Hasse's theorem, the probability that an $x \in \F_q$ lies on $\cale$ is in fact approximately $\frac{1}{2}$.

\subsection{Elliptic Curve El Gamal}

We describe the Elliptic Curve El Gamal (ECEG) public key encryption. The public parameters are:
\begin{enumerate}
\item a finite field $\F_q$ and elliptic curve $\cale: y^2 = x^3 + Ax + B$ defined over $\F_q$.
\item A point $P \in \cale(\F_q)$. 
\end{enumerate}
The receiver chooses $a \in \F_q$ as their secret key and publishes $aP$. For a sender to send a message $M$ to the receiver, they do the following:
\begin{enumerate}
\item choose $k \in \F_q$ and compute $k(aP)$.
\item Compute $Q = kP$ and $R = M + k(aP)$.
\item Send $(kP, R)$ to the receiver.
\end{enumerate}
Upon receipt, the receiver simply computer $M = R - a(kP)$. Correctness is immediate. In addition, ECEG is CPA-secure under the elliptic curve decisional Diffie-Hellman assumption.

In its original form, ECEG has a 1:2 expansion of plaintext to ciphertext~\cite{menezes}. However, to save space, the sender instead of sending $Q = (x_Q, y_Q)$ and $R = (x_R, y_R)$ can simply send $(x_Q, b_Q)$ and $(x_R, b_R)$, where $b_i$ are binary values denoting which square root the $y_i$ is, to save space.  This also eliminates an adversary's strategy of spotting distress signals by checking if, given $(X_P, Y_P)$ whether $Y_P^2 = X_P^3 + AX_P + B$ and if so, then with very high probability a distress has been sent.

\section{Demonstrating Adversary Advantage} \label{appendix:adversaryadvantage}

Consider the adversary $\adv$ with the following strategy. Given a distribution of nonces $r_1, \hdots\ r_n$:
\begin{enumerate}
\item for each $r_i$, check if $r_i$ has the ciphertext structure. 
\item If there exists a nonce that follows the ciphertext structure, then flip a coin.
\item If there are no nonces that follow the ciphertext structure, then output $b' =  0$.
\end{enumerate}

\begin{prop} Let $\adv$ be an adversary described above. Then $\IndAdv(n)$ is negligible.
\end{prop}

\begin{proof}
We consider an adversary $\adv$ that uses knowledge of the elliptic curve structure to distinguish between whether or not a distribution of nonces contain a distress call. The adversary's distinguishing advantage is:
\begin{align*}
\IndAdv & = | \Pr (b' = 1 | b = 1) - \Pr(b' = 1| b = 0)|.
\end{align*}

Let $D_b$ be the distribution of nonces that the adversary receives; that is, when $b = 1$ there exists a distress signal in the distribution. Let $r_1, \hdots, r_n$ be the nonces that the adversary receives.  Write $r_i = X_{i,1} || b_{i,1} || X_{i,2} || b_{i,2}$, where $X_{i,j}$ is 127 bits long and $b_{i,j}$ are single bits. Now, if $r_i$ is an output of mECEG, then $X_{i,j}$ would be on the curve $\cale_{A,B}: Y_{i,j}^2 = X_{i,j}^3 + AX_{i,j} + B$ for some $Y_{i,j}$. That is, $X_{i,j}^3 + AX_{i,j} + B$ is a quadratic residue. 
From Hasse's theorem, this has probability approximately $1/2$. 

For a nonce $r_i$, let $ec(r_i) \in \{0,1\}$ denote whether $r_i$ has the structure of an mECEG ciphertext, that is, $ec(r_i) = 1$ if and only if both $X_{i,1}$ and $X_{i,2}$ lie on the curve $\cale_{A,B}$.

Now, when the adversary receives $D_1$, the distribution may come from either $D_1$ or $D_0$, so $\Pr(b'=1) = \frac{1}{2}$. When the adversary receives $D_0$, then the adversary looks into each nonce. If $ec(r_i) = 1$, then the adversary flips a coin and if $ec(r_i) = 0$ then the adversary outputs $b=0$. Hence,
\begin{align*}
\Pr (b' = 1| b=0) = \frac{1}{2}\Pr(\exists r_i \text{s.t. } ec(r_i) = 1) + 0\cdot\Pr (\forall r_i \text{s.t. } ec(r_i) = 0).
\end{align*}

Therefore,
$\IndAdv =  | \frac{1}{2} - \frac{1}{2}\Pr(\exists r_i \text{s.t. } ec(r_i) = 1)|.$
Now, let $A_i$ be the event that $ec(r_i) = 1$. Let $\Pr(ec)(n)$ be the probability that there exists an $r_i$ with $i \in \{ 1, \hdots, n\}$ such that it has the mECEG ciphertext structure, so by the inclusion-exclusion principle we have
\begin{align*}
\Pr(ec)(n)  = \Pr \bigcup_{i = 1}^n A_i  =\sum_{k=1}^n (-1)^{k+1} \left( \sum_{1 \leq i_1 < \hdots < i_k \leq n} \Pr\left(A_{i_1} \cap \hdots \cap A_{i_k}\right)\right)
\end{align*}

Now $\Pr(A_i)$ is the event where both $X_{i,1}$ and $X_{i,2}$ lie on the curve, that is, $X_{i,j}^3 + AX_{i,j} + B$ is a square, with probability 1/2. 

For any $i$, we have $\Pr(A_i) = \frac{2^{2n-2}}{2^{2n}}$. The denominator comes from the fact that given $n$ nonces we are looking at whether \emph{both} $X_{i,1}$ and $X_{i,2}$ are on $\cale_{A,B}$, and the numerator looks at all the possible combinations (as the distribution of outputs $PRG$ is indistinguishable from the uniform distribution) conditioning on both $X_{i,1}$ and $X_{i,2}$ being squares.

Similarly, for $i\neq j \neq k$, we have $\Pr(A_i \cap A_j) = \frac{2^{2n-4}}{2^{2n}}$ and $\Pr(A_i \cap A_j \cap A_k) = \frac{2^{2n-6}}{2^{2n}}$, and so forth. Hence the adversary's advantage is

\begin{align*}
\IndAdv & = \left| \frac{1}{2} - \frac{1}{2}\left(\sum_{i=1}^n {n \choose i} (-1)^{i+1} \frac{2^{2n-2i)}}{2^{2n}}\right)\right| \\
& = \left| \frac{1}{2} - \frac{1}{2}\left(\sum_{i=1}^n {n \choose i} (-1)^{i+1} \frac{1}{2^{2i}}\right)\right|, 
\end{align*}
which is negligible (and is shown in Figure~\ref{figure:advantage}).
\end{proof}

\begin{figure}[tp]
\centering
\includegraphics[width=0.9\linewidth]{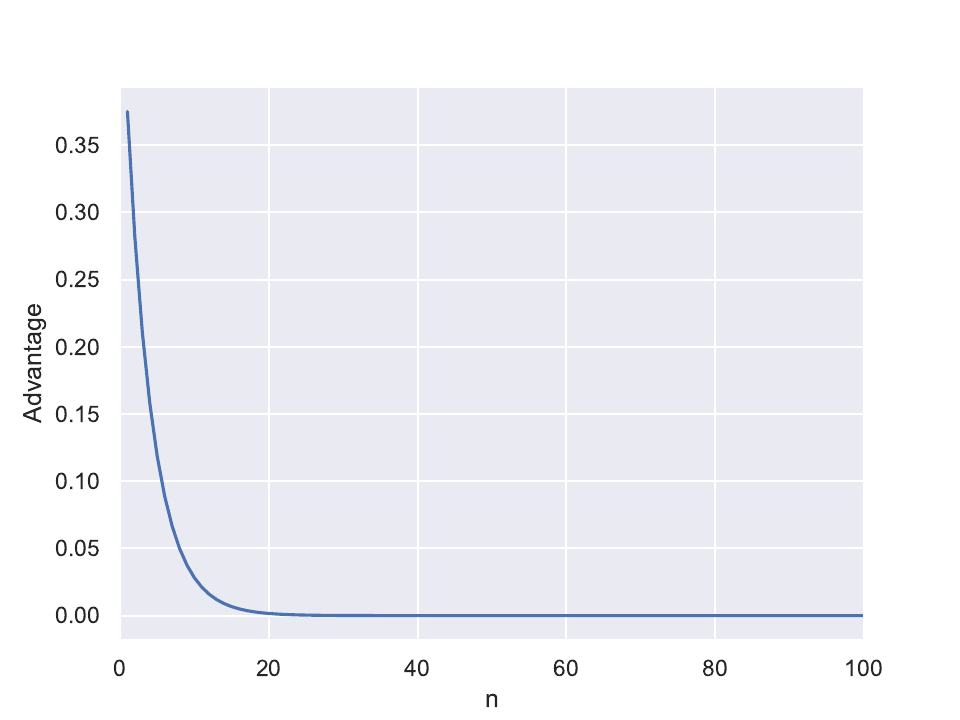}
\caption{$\IndAdv(n)$ quickly goes to 0 as $n$ increases.}
\label{figure:advantage}
\end{figure}

\end{document}